\documentclass[runningheads]{llncs}

\usepackage[utf8]{inputenc}
\usepackage{algorithmicx}
\usepackage{algorithm}
\usepackage{algpseudocode}
\usepackage{amsmath}
\usepackage{graphicx}
\usepackage{amssymb}
\usepackage{color}
\usepackage{hyperref}

\newtheorem{clm}{Claim}
\newtheorem{thm1}{Theorem}

\DeclareMathOperator{\depth}{depth}
\DeclareMathOperator{\cost}{cost}
\DeclareMathOperator{\leftn}{left()}
\DeclareMathOperator{\rightn}{right()}
\DeclareMathOperator{\rootn}{root()}
\DeclareMathOperator{\stat}{STAT}
\def\A{\mathcal{A}}
\def\On{\textsc{On}}

\title{Arithmetic Binary Search Trees: \\
Static Optimality in the Matching Model}
\author{Chen Avin}
\institute{School of Electrical and Computer Engineering \\
Ben Gurion University of the Negev, Israel\\
\email{avin@cse.bgu.ac.il}}

\begin{document}
\maketitle

\begin{abstract}
Motivated by recent developments in optical switching and reconfigurable network design, we study dynamic binary search trees (BSTs) in the matching model. In the classical dynamic BST model, the cost of both link traversal and basic reconfiguration (rotation) is $O(1)$. However, in the matching model, the BST is defined by two optical switches (that represent two matchings in an abstract way), and each switch (or matching) reconfiguration cost is $\alpha$ while a link traversal cost is still $O(1)$. 
In this work, we propose Arithmetic BST (A-BST), a simple dynamic BST algorithm that is based on dynamic Shannon-Fano-Elias coding, and show that A-BST is statically optimal for sequences of length $\Omega(n \alpha \log \alpha)$ where $n$ is the number of nodes (keys) in the tree.

\keywords{binary search trees \and
static optimality \and
matching model \and
arithmetic coding \and
Entropy.}
\end{abstract}

\section{Introduction}

In this paper, we study one of the most classical problems in computer science, the design of efficient binary search trees (BSTs) \cite{knuth1971optimum}.  More concretely, we are interested in \emph{Dynamic} \cite{mehlhorn1979dynamic} or \emph{Self-Adjusting} \cite{sleator1985self} binary search trees that need to serve a sequence of requests. 
While traditionally binary search trees are studied in the context of data-structures with a focus on memory and running time optimization, we are motivated in the physical world implementations of binary search trees and, in particular, \emph{self-adjusting} networks \cite{avin2019toward}.

\paragraph{\bf The Matching model and Self-adjusting networks.} Self-adjusting networks, or reconfigurable networks, are communication networks that enable dynamic physical layer topologies \cite{farrington2010helios,ghobadi2016projector,mellette2017rotornet,ballani2020sirius}, i.e.,  topologies that can change over time. Recently, reconfigurable optical switching technologies have introduced an intriguing alternative
to the traditional design of datacenter networks, allowing to dynamically establish direct shortcuts between communicating partners, depending on the demand \cite{ghobadi2016projector,hamedazimi2014firefly,mellette2017rotornet,alistarh2015high}. For example, such reconfigurable links could be established to support elephant flows or traffic between two racks with significant communication demands. The potential for such demand-aware optimizations is high: empirical studies show that communication traffic features both spatial and temporal \emph{locality}, i.e., 
traffic matrices are indeed sparse and a small number of elephant flows can constitute a significant fraction of the datacenter traffic \cite{roy2015inside,benson2009understanding,avin2020complexity}.
The main metric of interest in these networks is the flow completion time or the average packet delay, where the reconfiguration delay (aka as latency tax \cite{griner2020performance}) can be several orders of magnitude larger than the forwarding delay, i.e., milliseconds or microseconds vs. nanoseconds or less. This brings a striking contrast to the common data-structures approach where pointer forwarding and pointer changing (e.g., rotations) are of the \emph{same} time order and considered as a unit cost.

Previous work \cite{schmid2015splaynet,avin2019renets,avin2017demand} has shown that binary search trees (and other types of trees) can be used as an important building block in the design of self-adjusting networks since they carry nice properties like local routing and low degree . Moreover, to capture many recent designs of demand-aware networks like \cite{ghobadi2016projector,mellette2017rotornet,ballani2020sirius,farrington2010helios,xpander}, a simple leaf-spine datacenter network model called ToR-Matching-ToR (TMT) was recently proposed \cite{avin2020online}. In the TMT model, $n$ Top of the Rack (ToR) switches (leaves) are connected via optical switches (spines), and each spine switch can have a dynamically changing \emph{matching} between its $n$ input-output ports. See Figure \ref{fig:system} for an illustration. The exact networking and technical details of the model are out of the scope of the paper, but abstractly with only two spine switches (i.e., two matchings), the model supports a simple implementation of a dynamic binary search tree that enables a greedy routing from the root (a ToR) to any other node (ToR). 

\begin{figure}[t!]
\begin{centering}
\includegraphics[width=.6\columnwidth]{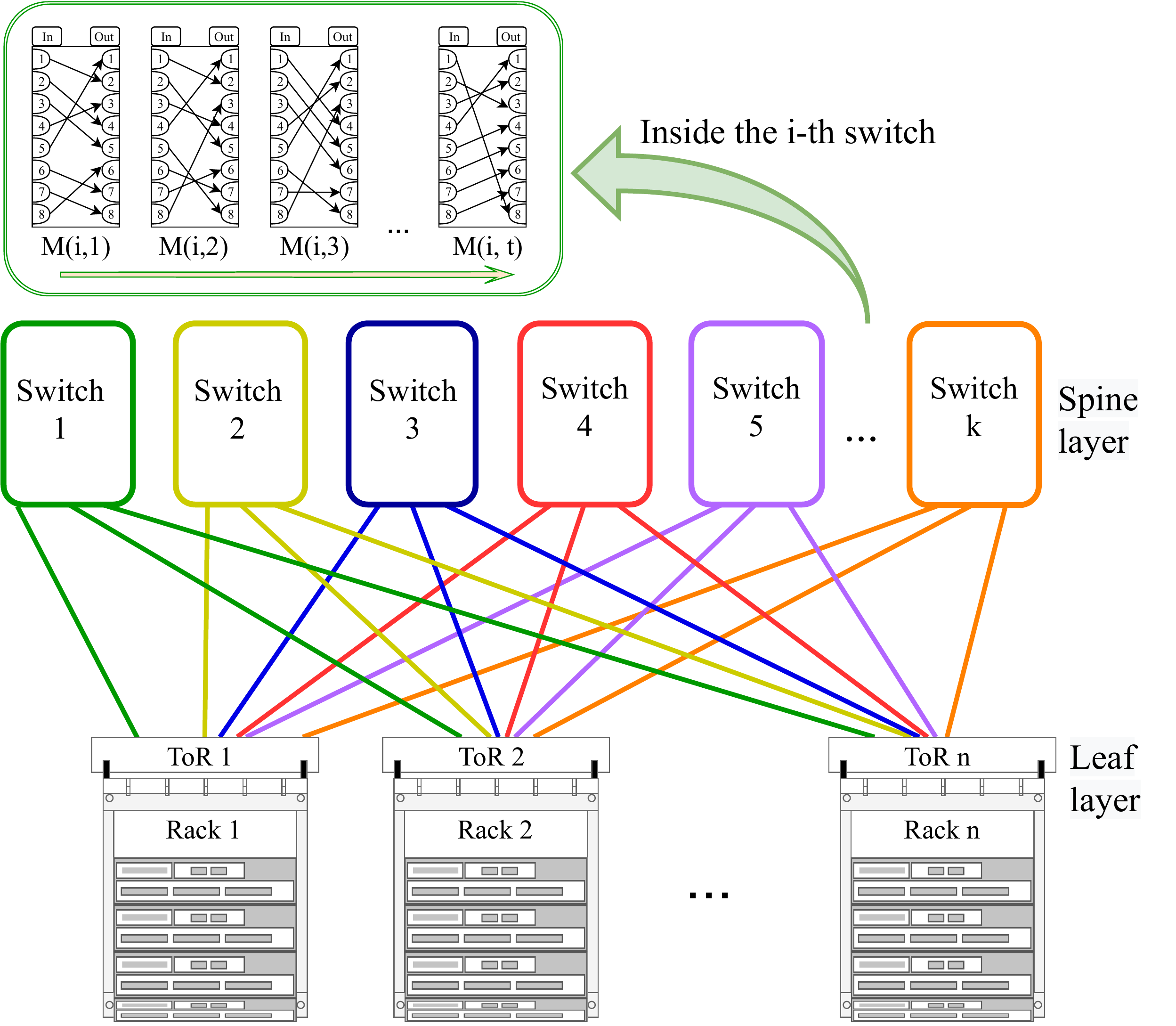}     
\caption{Overview of ToR-Matching-ToR (TMT) network model \cite{avin2020online}. The $n$ nodes of the BST are the leaf switches and they are interconnected via $n$ port spine switches, each consist of a (dynamic) matching. A BST can be implemented in this model using two spine switches. The cost (reconfiguration delay) of updating a matching is $\alpha$.}
\label{fig:system}
\end{centering}
\end{figure}

This leads to the \emph{matching model} for dynamic binary search trees that we study in this paper where the tree can change at any time to any other tree but at a reconfiguration cost which is a parameter denoted as $\alpha$. 

\subsection{Formal Model and Main Result}

We consider a dynamic binary search tree (BST) that needs to serve a sequence of requests from a set $V$ of $n$ unique keys, where the value of the $i$'th key is $v_i$ and the keys impose a complete order.
For simplicity and w.l.o.g  we assume that the keys are sorted by their value, i.e., for $i<j$, $v_i \le v_j$ (otherwise we can just rename the keys). 
Let $\sigma = \{\sigma_1, \sigma_2, \dots \sigma_m\}$ be a sequence of $m$ requests for keys where $\sigma_j \in V$.

A binary tree $T$ has a finite set of nodes with one node designated to be the \emph{root}. Each node has a \emph{left} and a \emph{right} child which can be either a different node or a null object. Every node in $T$, but for the root node, has a single parent node for which it is a child. 
We denote the root of a tree $T$ by $T.\rootn$ and for a node $v$, let $v.\leftn$ and $v.\rightn$ denote it left and right children. 
The depth of $v$ in $T$ is defined as the number of nodes on the path from the root to $v$, denoted as $\depth(T, v)$. The depth of the root is therefore one.

In a \emph{binary search tree} $B$, each node has a unique key. We assume the set of nodes to be $V$ and use nodes and keys interchangeably. The tree satisfies the \emph{symmetric order} condition: every node’s key is greater than the keys in its left child subtree and smaller than the keys in its right child subtree.

For a given set of keys $V$ and a sequence $\sigma$, the cost of the optimal \emph{static} BST is defined as the BST with minimal cost to serve $\sigma$, formally
\begin{align}
    \stat(\sigma) = \min_{B \in \mathcal{B}} \sum_{t=1}^m \depth(B,\sigma_t)
\end{align}
where $\mathcal{B}$ is the set of all binary search trees.

A \emph{dynamic} BST is a sequence of BSTs, $B_t$ with $V$ as the set of nodes. 
At each time step $t$, the cost of searching (or serving) a key $\sigma_t$ is 
$\depth(B_t, \sigma_t)$, but after the request is served the BST $B_t$ can adjust its structure (while maintaining the symmetric order property). 

The key point of this paper and in our model is that the cost assumptions differ.
In previous studies, classical  \cite{baer1975weight,allen1978self,mehlhorn1979dynamic,sleator1985self}, as well as more recent ones \cite{demaine2007dynamic,levy2019new,chalermsook2020new}, to name a few, the cost of both a single link traversal and a single pointer change (usually called rotation) is $O(1)$ units. 
In most of these algorithms the cost of reconfiguration of the tree is kept proportional to cost of accessing the recent key.
In contrast, in the \emph{matching model} the cost model assumes that \emph{any} change from $B_t$ to $B_{t+1}$ is possible (it could be a single rotation or a completely new tree) but at the cost of $\alpha$ units.

For a dynamic BST algorithm $\A$ the total cost of serving $\sigma$ starting from a BST $B=B_1$ is defined in the matching model as follows,
\begin{align}
    \cost(\A, \sigma, B) = \sum_{t=1}^m \depth(B_t,\sigma_t) + \alpha \mathbb{I}_{B_t \neq B_{t+1}}
\end{align}
Where $\mathbb{I}_e$ is the indicator function, i.e., $\mathbb{I}_e=1$ if the event $e$ is true and 0 otherwise.
From a networking perspective the cost of a dynamic BST algorithm can be seen as the total delay to serve (i.e., route) requests to nodes in $\sigma$ from a single source that is connected to the tree's root. 

In \emph{static optimality}, which we next define, we want the dynamic BST algorithm to (asymptotically) perform 
well \emph{even in hindsight} compared to the best static tree. 

\begin{definition}[\bf Static Optimality]
Let $\stat$ be the cost of an optimal static BST with perfect
knowledge of the demand~$\sigma$, and let~$\On$ be an online algorithm for dynamic BST. 
We say that~$\On$ is \emph{statically optimal} if there exists some constant $\rho \ge 1$, and for any sufficiently long sequence of keys $\sigma$, we have that
$$
\cost(\On, \sigma, B) \le \rho \stat(\sigma)
$$
\noindent
where $B$ is the initial BST from which $\On$ starts.
In other words,~$\On$'s cost is at most a 
constant factor higher than~$\stat$ in the worst case.
\end{definition}

Note that when $\alpha=O(1)$ known algorithms for dynamic BST that are static optimal, e.g.,  \cite{baer1975weight,mehlhorn1979dynamic,sleator1985self,demaine2007dynamic,demaine2009geometry}, can be used to achieve static optimality in our model. The more interesting scenario is when $\alpha=\omega(1)$ and naive current algorithms with rotation cost (or splay to the root, or move to root) of $\alpha$ will not be static optimal. In this work, we assume for simplicity that $\alpha \ge 2$.

Let $w_i(t)$ be the number of appearances of key $i$ in $\sigma$ up to time $t$. Note that $\sum_{i=1}^n w_i(t)=t$. For short let $w_i = w_i(m)$.
Let $\overline{W}=\overline{W}(\sigma)$ be the frequency (or empirical) distribution of $\sigma$, i.e., $\overline{W}=\{\frac{w_1}{m}, \frac{w_2}{m}, \dots, \frac{w_n}{m}\}$, and $\overline{W}(t)$ the frequency distribution up to time $t$.
It is a classical result that the optimal \emph{static} BST for $\sigma$ has amortized access cost of $\Theta(H(\overline{W}))$ where $H$ is the entropy function \cite{mehlhorn1975nearly}, i.e.,
$\stat(\sigma) = \Theta(m (1+H(\overline{W}(\sigma))))$.

The main result of this paper is a simple dynamic BST algorithm, A-BST, that is based on arithmetic coding \cite{Witten87arithmetic} and in particular on a dynamic Shannon-Fano-Elias coding \cite{cover2012elements} and is statically optimal for sequences of length $\Omega(n \alpha \log \alpha)$. Formally,

\begin{thm1}
\label{thm:main}
A-BST (Algorithm \ref{alg:A-BST}) is a statically optimal, dynamic BST for sequences of length at least  $2 n \alpha \log \alpha$.
\end{thm1}

The rest of the paper is organized as follows: In Section \ref{sec:pre} we review related concepts like Entropy and the Shannon-Fano-Elias (SFE) coding with a detailed example. Section \ref{sec:algo} presents A-BST including  related algorithms and an example for its dynamic operation. In Section \ref{sec:proof} we prove Theorem \ref{thm:main}. Finally, we conclude with a discussion and open questions in Section \ref{sec:discussion}.

\section{Preliminaries}\label{sec:pre}
\paragraph{Entropy and Shannon-Fano-Elias (SFE) Coding.}

Entropy is a known measure of unpredictability of information content~\cite{shannon1948mathematical}. For a discrete random variable $X$ with possible values
$\{x_1, \dots , x_n\}$, the (binary) entropy $H(X)$ of $X$ is defined as
$H(X) = \sum_{i=1}^n p_i\log_2\frac{1}{p_i}$
where $p_i$ is the probability that $X$ takes the value $x_i$. Note that, $0 \cdot \log_2\frac{1}{0}=0$ and we usually assume that
$p_i > 0$ $\forall i$. Let $\overline{p}$ denote $X$'s probability distribution,
 then we may write $H(\overline{p})$ instead of $H(X)$.

Shannon-Fano-Elias (SFE)~\cite{cover2012elements} is a well-known symbol, \emph{prefix code} for lossless data compression and is the predecessor of the more famous and used (variable block length) arithmetic coding \cite{Witten87arithmetic}. 
The Shannon-Fano-Elias code produce variable-length code-words
based on the  probability of each possible symbol.
As in other entropy encoding methods (e.g., Huffman), higher probability symbols are represented using fewer bits, to reduce the expected code length. Unlike Huffman coding the coding method works for any order of symbols and the probabilities do not need to be sorted. This will fit better with the searching property we need in the tree later on.  

The encoding is based on the cumulative distribution function (CDF) of the probability distribution, 
\begin{align} \label{eq:cdf}
F(i) = \sum\limits_{j \leq i} p_j,
\end{align}
%
and encodes symbols using the function $\overline{F}$ where
\begin{align}\label{eq:cwfunc}
	\overline{F}(i) = \sum_{j < i} p_j + \frac{p_i}{2} = F(i-1) + \frac{p_i}{2}. 
\end{align}
%
%
Denote by $B(i)$ the binary representation of $\overline{F}(i)$. 
The code-word $C(i)$ for the $i$'th symbol $x_i$ consists of the first $\ell(i)$ bits of the fractional part of $B(i)$, denoted as
\begin{align}\label{eq:cword}
	C(i) = \lfloor B(i) \rfloor_{\ell(i)}, 
\end{align}
where the \emph{code length} $\ell(i)$ is defined as 
\begin{align}\label{eq:cwlength}
\ell_i = \left \lceil \log \frac{1}{p_i} \right \rceil + 1. 
\end{align}

%
%

The above construction guarantees $(i)$ that the code-words $C(i)$ are prefix-free 
and therefore appear as leaves in the binary tree that represent the code, and $(ii)$ that the average code length $L_{\mathrm{SFE}}(X)$, 
\begin{align}\label{eq:ecl}
L_{\mathrm{SFE}}(X) = \sum\limits_{i = 1}^n p_i \cdot l(i) &= \sum\limits_{i = 1}^n p_i (\lceil \log \frac{1}{p_i} \rceil + 1) 
\end{align}
is close to the entropy $H(X)$ of the random variable $X$, and in particular, 
\begin{align}\label{eq:eclbound}
	H(X) + 1 \leq L_{\mathrm{SFE}}(X) < H(X) + 2,
\end{align}

\begin{figure}[t!]
    \centering
    \begin{tabular}{c|c}
         \emph{\hypertarget{ExA}{Example~A}} & \emph{\hypertarget{ExB}{Example~B}} \\ \\
         \includegraphics[height=2cm]{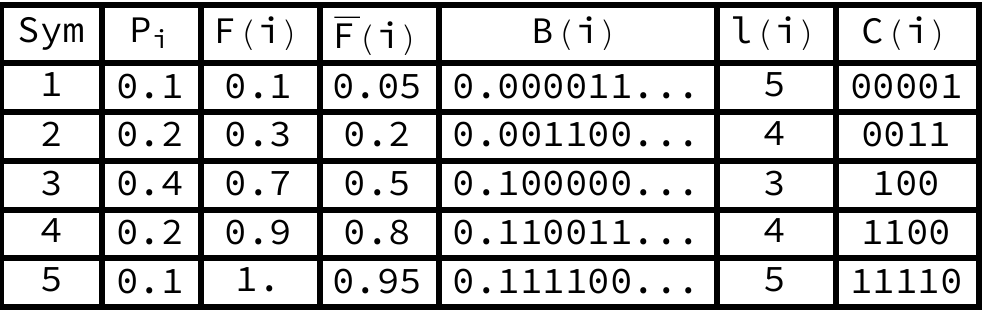} & \includegraphics[height=2cm]{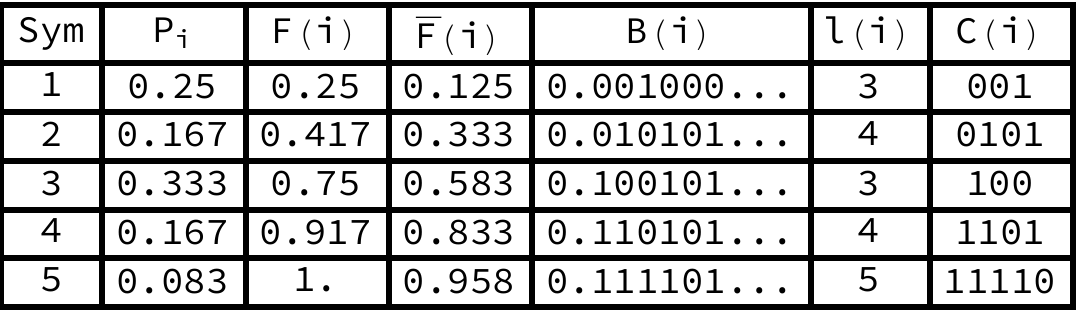} \\
         (a) & (e)\\ \\
         \includegraphics[height=3.7cm]{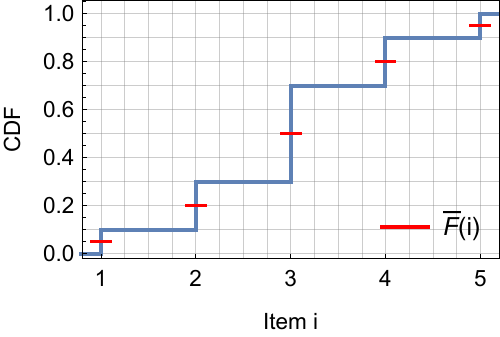} & \includegraphics[height=3.7cm]{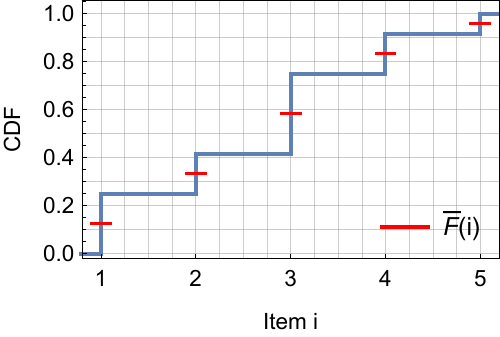} \\
         (b) & (f)\\ \\
         \includegraphics[height=3.5cm]{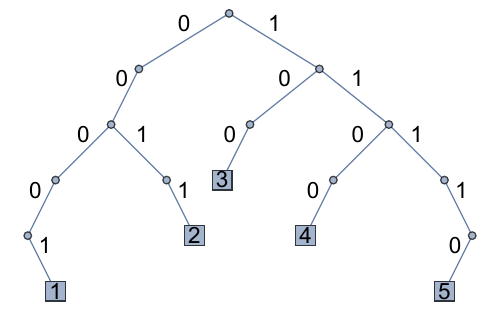} & \includegraphics[height=3.5cm]{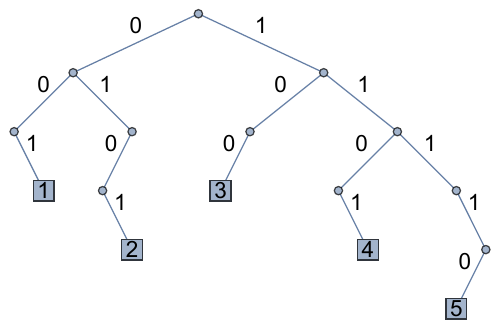} \\
         (c) & (g)\\ \\
         \includegraphics[height=3cm]{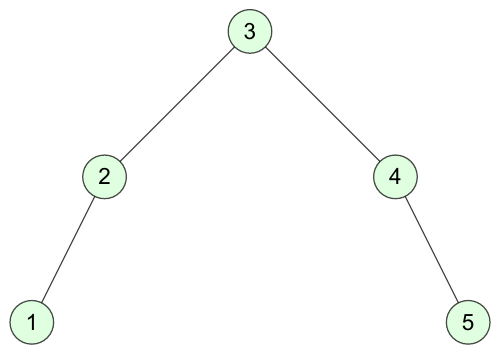} & \includegraphics[height=3cm]{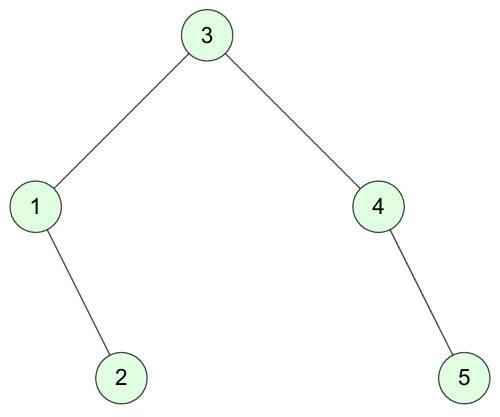} \\
         (d) & (h)
         
    \end{tabular}
    \caption{Our two running examples A and B. The top row shows the SFE coding details, second row shows the functions $F(i)$ and $\overline{F}(i)$, the third row shows the SFE prefix-free code-words as a binary tree and the last row present the corresponding BST after running algorithm \ref{alg:P2B}: PrefixTree2BST($T$).}
    \label{fig:example1}
\end{figure}

For an illustration consider \emph{\hyperlink{ExA}{Example~A}} in Figure \ref{fig:example1} which we will use also later in the paper. Let $X$ be a random variable with five possible symbols $\{1, 2, 3, 4, 5\}$ and the corresponding probabilities  $\{0.1, 0.2, 0.4, 0.2, 0.1\}$.
Figure \ref{fig:example1} (a) provides the detailed parameters for the encoding of each symbol $i$, including $F(i), \overline{F}(i)$, the binary representation $B(i)$, the code-word length $\ell(i)$ and the (prefix-free) binary code-word $C(i)$. Additionally, Figure \ref{fig:example1} (b) presents $F(i)$ and $\overline{F}(i)$ graphically and Figure \ref{fig:example1} (c) shows a binary tree with leaves as the code-words where the path from the root to each leaf represent its (prefix-free) code-word. Figure \ref{fig:example1} (d) shows the conversion of the prefix tree to a binary search tree which we will discuss next. 

\section{Arithmetic Binary Search Trees (A-BST)}\label{sec:algo}

The idea of arithmetic binary search trees is simple in principle. It is composed of three phases:
\begin{enumerate}
    \item Use arithmetic coding, and in particular, Shannon-Fano-Elias coding to create an efficient, variable-length, prefix-free code for a given (empirical) distribution.
    \item In turn, convert the binary tree of the prefix-free code to a \emph{biased} binary search tree (with an entropy bound on its performance). 
    \item Dynamically update the empirical distribution, the prefix-free code and the corresponding binary search tree.  
\end{enumerate}

We first describe steps $1$ and $2$ that may be of independent interest and then the crux of the method which is the dynamic update process. 

\subsection{From Shannon-Fano-Elias (SFE) coding to BST}

\begin{algorithm}[t]
    \caption{SFE-2-BST($\overline{P}$): Convert Arithmetic (SFE) coding to BST
        \label{alg:convert}}
    \begin{algorithmic}[1]
        \Require A probability distribution $\overline{P}$. $p_i$ is the probability of key $i$, the key with rank $i$ in the sorted keys' list.  
        \Ensure A BST $B$ where key $i$ is at distance $O(\log \frac{1}{p_i})$ from the root.
        \State Create a prefix-free code-word for each key $i$ via SFE coding. 
        \State Create a binary tree $T$ from the SFE-code, with keys as leaves (i.e., binary trie). 
        \State $B$ = PrefixTree-2-BST($T$). \Comment{See Algorithm \ref{alg:P2B}}
    \end{algorithmic}
\end{algorithm}

Algorithms \ref{alg:convert} and \ref{alg:P2B} shows how to create a near optimal biased BST for a given distribution via SFE coding.
Algorithm SFE-2-BST (Algorithm \ref{alg:convert}) is first provided with a probability distribution $\overline{P}$ for the biased BST. Note that, w.l.o.g, the distribution is sorted by the keys values. From $\overline{P}$ the algorithm then creates a \emph{prefix-free} binary code using SFE coding\footnote{Any other coding method that preserves the keys order can be used, e.g., Shannon–Fano coding \cite{cover2012elements} which split the probabilities to half, Alphabetic Codes \cite{yeung1991alphabetic} or Mehlhorn tree \cite{mehlhorn1975nearly}, but not e.g., Huffman coding \cite{huffman1952method}, that needs to sort the probabilities.} and the 
corresponding binary tree $T$ where keys (symbols) are at the leaves and ordered by their value (and not by probabilities). 
Next, the algorithm calls Algorithm PrefixTree-2-BST (Algorithm \ref{alg:P2B}) which converts $T$ to a BST $B$.

\begin{algorithm}[t]
    \caption{PrefixTree-2-BST($T$):  Convert a binary prefix tree to BST
        \label{alg:P2B}}
    \begin{algorithmic}[1]
        \Require A prefix tree $T$ with keys as leaves in sorted order.
        \Ensure A BST $B$ where each key in $B$ is at least as close to the root as in $T$.
        \If{$T$ is of size one or NULL}
            \State return $T$
        \Else
            \State Let $pre$ be the rightmost leaf in the left subtree of $T$ (pre-order)
            \State Let $post$ the leftmost leaf in the right subtree of $T$ (post-order)
            \State Let $\ell$ be the leaf with lower depth between $pre$ and $post$
            \State Let $B$ be a new binary search tree with $\ell$ as a root
            \State Delete $\ell$ from $T$
            \State $B$.root.left$ = $PrefixTree-2-BST($T$.root.left)
            \State $B$.root.right$ = $PrefixTree-2-BST($T$.root.right)
            \State return $B$
        \EndIf
    \end{algorithmic}
\end{algorithm}

Algorithm \ref{alg:P2B} is a recursive algorithm that receives a binary tree $T$ with keys as leaves in increasing order. It creates a BST $B$ by setting the root of the tree as the key which is closer to the root between the two keys that are pre-order and post-order to the root (and then deleting the key from $T$). The left and right subtrees of $B$'s root are created by calling recursively the algorithm with the relevant subtrees of $T$.

Recall Figure \ref{fig:example1} (c) which shows a binary tree $T$ for the SFE code of the probability distribution in \emph{\hyperlink{ExA}{Example~A}}. The keys $1,2,3,4,5$ are leaves and in sorted order.  This tree is an intermediate result of SFE-2-BST($\overline{P}$) algorithm. Figure \ref{fig:example1} (d) presents the corresponding BST after calling PrefixTree-2-BST$(T)$.
Initially $3$ is selected as the root as it is closer to the root between keys $2$ and $3$. Then (after deleting $3$ from the tree), the algorithm continues recursively, by building the children of $B$'s root with the left and right subtrees of $T$'s root. On the left, 2 is selected as the next root and then 1 as its left child. On the right $4$ is selected as the next root and then 5 as its right child.

We can easily bound the depth of any key in the tree generated by SFE-2-BST.

\begin{clm}\label{clm:depth}
For a key-sorted probability distribution $\overline{P}$, the algorithm SFE-2-BST($\overline{P}$) creates a biased BST $B$ where for each key $i$ it holds $\depth(B, i) < \log \frac{1}{p_i} +3$.
\end{clm}
\begin{proof}
The SFE code creates code-words with length $\ell_i = \log \frac{1}{p_i} + 2$, so the tree $T$ in Algorithm SFE-2-BST has the leaves at depth $\ell_i +1$. Algorithm PrefixTree-2-BST($T$) does not change the symmetric order of keys that are initially stored at leaves, and can only decrease the depth of any key. \qed
\end{proof}

We next discuss how to make the BST dynamic.

\subsection{Dynamic BST}

Algorithm \ref{alg:A-BST} provides a pseudo-code for a dynamic Arithmetic Binary Search Tree (A-BST) $B_t$.
The algorithm starts from $B_1$ which is a balanced BST over the $n$ possible keys (assuming no knowledge on the keys distribution). 
The algorithm maintains two probability distribution vectors  that change over time, $\overline{P}=\{p_1, \dots , p_n\}$ and $\overline{Q}=\{q_1, \dots , q_n\}$. $\overline{P}$ holds the probability distribution by which the current BST was built and initially is set to the uniform distribution. $\overline{Q}$ holds the empirical distribution of $\sigma$ up to time $t$ and is updated on each request, i.e., $q_i(t) = \frac{w_i(t)}{t}$~\footnote{To avoid zero probabilities, in practice $q_i(t)$ is set by Laplace rule \cite{carnap1947application} to be $\frac{w_i(t)+1}{t+n}$, this does not change the main Theorem since for $t \ge n$ we have $\frac{w_i(t)+1}{t+n} \ge \frac{w_i(t)}{2t}$. For simplicity of exposition we assume $q_i(t) = \frac{w_i(t)}{t}$.}.
The crux of the algorithm is that, whenever a key $i$ is requested, if as a result $p_i < q_i/2$, then  
$\overline{P}$ is updated to be equal to $\overline{Q}$, and 
a new BST is created (at cost of $\alpha$) from $\overline{P}$.
Note that upon a request to key $i$ only $q_i$ increases while all other probabilities decrease. So the condition $p_j < q_j/2$ can only become valid for $j=i$ when key $i$ is requested.

We demonstrate the algorithm using \emph{\hyperlink{ExB}{Example~B}}. Assume we have $5$ keys $1,2,3,4,5$ and after the first ten requests in $\sigma$ we have $\overline{Q} = \overline{P} = \{ \frac{1}{10}, \frac{2}{10}, \frac{4}{10}, \frac{2}{10}, \frac{1}{10}\}$. Note that this is the empirical distribution of \hyperlink{ExA}{Example~A}, so $B_{10}$ is the BST that is shown in Figure \ref{fig:example1} (d). Next we assume that the $11th$ and $12th$ requests are for key 1. 
After the $11th$ request $p_1(11) = \frac{1}{10}$ and $q_1(11) = \frac{2}{11}$ so $p_1(11) > q_i(11)/2$ and $B_{11}$ remains as $B_{10}$. After the $12th$ request we have $p_1(12) = \frac{1}{10}$ and $q_1(12) = \frac{3}{12}$, and now  $p_1(12) < q_i(12)/2$. So at time $t=12$, a new SFE code and a new BST are created with the empirical distribution $\overline{Q} = \overline{P} = \frac{3}{12}, \frac{2}{12}, \frac{4}{12}, \frac{2}{12}, \frac{1}{12}$. The details of the code for this distribution are given in Figure \ref{fig:example1}~(e)-(g) and the new BST, $B_{12}$ is shown in Figure \ref{fig:example1}~(h). Note that key 1 got closer to the root.

Next we analyse A-BST and show it is statically optimal.

\begin{algorithm}[t]
    \caption{A-BST: Simple Arithmetic BST
        \label{alg:A-BST}}
    \begin{algorithmic}[1]
        \Require A BST $B_t$, $\overline{P}$ - the current tree distribution, $\overline{Q}$ - empirical distribution (or model distribution). $\forall i, ~ p_i \ge q_i/2$.      
        \Ensure Dynamic BST $B_{t+1}$, $\forall i, ~ p_i \ge q_i/2$.
        \State Upon request of key $i$ at time $t$
        \State Update $\overline{Q}$  \Comment{$q_i(t) = \frac{w_i(t)}{t}=\frac{w_i(t-1)+1}{t}$}
        \If{$p_i < q_i/2$ } \Comment{update $\overline{P}$ if needed}
            \State Set $\overline{P}$ to $\overline{Q}$ \label{ln:set}
            \State $B_{t+1} = $ SFE-2-BST($\overline{P}$) \label{ln:B }\Comment{at cost $O(\alpha)$}
        \Else 
            \State $B_{t+1}=B_{t}$
        \EndIf
        \State Serve key $i$ on $B_{t+1}$ at cost $O(\log \frac{1}{p_i})$
    \end{algorithmic}
\end{algorithm}

\section{A-BST is statically Optimal}\label{sec:proof}
We now prove the main result of this work. Recall that we assume $\alpha \ge 2$ and the more intresting case is $\alpha=\omega(1)$. 
\begin{theorem}
A-BST (Algorithm \ref{alg:A-BST}) is a statically optimal, dynamic BST for sequences of length at least  $2 n \alpha \log \alpha$.
\end{theorem}

\begin{proof}
Let $B_t$ be the BST tree of the algorithm at time $t$. Let $\overline{P}_t = \{p_1(t), \dots, p_n(t)\}$ be the probability distribution by which $B_t$ was build. 
Let $\overline{Q}_t= \{q_1(t), \dots, q_n(t)\}$ be the frequency distribution of $\sigma$ up to time $t$.
Note that $q_i(m) = \frac{w_i}{m}$.
Observe that Algorithm \ref{alg:A-BST} guarantees that for each $i$ and each $t$ we have $p_i(t) \ge \frac{q_i(t)}{2}$.
We first analyse the \emph{searching cost} (or access cost \cite{blum2016static}), $\cost_S(\sigma)$ of Algorithm \ref{alg:A-BST}.

\begin{align}
    \cost_S(\sigma)&= \sum_{t=1}^m \depth(B_t, \sigma_t) \\
                &\le \sum_{t=1}^m \left ( \log(\frac{1}{p_{\sigma_t}(t)}) + 3 \right ) \label{eq:depthbound} \\
                &\le \sum_{t=1}^m \left ( \log(\frac{2}{q_{\sigma_t}(t)}) + 3 \right ) \\
                &\le  4m + \sum_{t=1}^m \log(\frac{1}{q_{\sigma_t}(t)})
\end{align}
where Eq. \eqref{eq:depthbound} is due to Claim \ref{clm:depth}.

Next we consider the cost $\sum_{t=1}^m \log(\frac{1}{q_{\sigma_t}(t)})$ and analyze 
each key $i$ at a time. We follow an idea mentioned in a classical paper on dynamic Huffman coding by Vitter \cite{vitter1987design} and is due to a personal communication of Vitter with B. Chazelle.
Consider the \emph{last} $\lfloor \frac{w_i}{2} \rfloor$ requests for key $i$. For each such request, since $w_i(t) \ge w_i/2$, we have $q_i(t) =\frac{w_i(t)}{t} \ge \frac{w_i}{2t} \ge \frac{w_i}{2m}$ . So the sum of these $\frac{w_i}{2}$ requests is at most $\frac{w_i}{2} \log \frac{2m}{w_i}$.

Similarly for each $j$'th request of key $i$ where $\lceil \frac{w_i}{2^{k}} \rceil< j \le \lceil \frac{w_i}{2^{k-1}} \rceil$ where $k=1, 2, \dots \lceil \log(w_i) \rceil$, 
$q_i(t) \ge \frac{w_i}{2^{k}m}$ and the cost of all of these requests is less than $\frac{w_i}{2^k} \log \frac{w_i}{2^k m}$.
So for each key $i$
\begin{align}
  \sum_{\sigma_t=i} \log(\frac{1}{q_{i}(t))}) &\le \sum_{k=1}^{\lceil \log(w_i) \rceil} \frac{w_i}{2^k} \log \frac{2^{k}m}{w_i} \\
  &\le \sum_{k=1}^{\lceil \log(w_i) \rceil} \frac{w_i}{2^k}( \log \frac{m}{w_i} + \log 2^{k} ) \\
  &\le w_i \log \frac{m}{w_i}\sum_{k=1}^{\lceil \log(w_i) \rceil} \frac{1}{2^k} + w_i \sum_{k=1}^{\lceil \log(w_i) \rceil} \frac{k}{2^k} \\
  &\le w_i \log \frac{m}{w_i} + 2 w_i
\end{align}
From this it follows that:
\begin{align}
    \sum_{t=1}^m \log(\frac{1}{q_{\sigma_t}(t))}) &\le \sum_{i=1}^n w_i \log \frac{m}{w_i} + 2 w_i \\
    &\le m H(\overline{W}) + 2m \label{eq:entropy}
\end{align}
A result identical to Eq. \eqref{eq:entropy} was recently shown in \cite{Golin18Dynamic}.
Next we consider the \emph{adjustment} cost. We do it again, one key at a time.
For a key $i$ consider all the adjustment it caused (lines 5-6 in Algorithm \ref{alg:A-BST}). 
We consider two cases: i) all adjustment when $w_i(t) \ge 2\alpha$ and ii) $w_i(t) < 2 \alpha$.
To prove these two cases we will use the following claim.
\begin{clm}\label{clm:double}
Let $i$ be the key that cause an adjustment (lines 5-6 in Algorithm \ref{alg:A-BST}) at time $t$.
Let $t'$ be the time of the previous adjustment initiated by any key. Then $w_i(t') < \frac{w_i(t)}{2}$.
\end{clm}
\begin{proof}
By Algorithm \ref{alg:A-BST} after the adjustment at time $t'$, 
$\overline{P}$ was set to $\overline{Q}$ and in particular $p_i(t')$ was set to $q_i(t')$.
The result follows since at time $t$ we have
\begin{align}
p_i(t) = p_i(t') = \frac{w_i(t')}{t'} &< \frac{q_i(t)}{2} = \frac{w_i(t)}{2t} \le \frac{w_i(t)}{2t'} 
\end{align}
\qed
\end{proof}
We can now analyse the two adjustment cases.
\paragraph{Case i: $w_i(t) \ge 2\alpha$.}
First notice that if $w_i(t') \ge 2\alpha$ then for each $t > t'$ $w_i(t) \ge 2\alpha$.
Assume a time $t$ for which $w_i(t) \ge 2\alpha$ and lines 5-6 in Algorithm \ref{alg:A-BST} was executed.
Now consider time $t' < t$ where the previous execution of lines 5-6 occurred. 

From Claim \ref{clm:double} and using the assumption that $w_i(t) \ge 2\alpha$, the number of requests to key $i$ since the last adjustment is:
\begin{align}
    w_i(t) - w_i(t') \ge w_i(t) - \frac{w_i(t)}{2} = \frac{w_i(t)}{2} \ge \alpha
\end{align}
Therefore we can amortize the adjustment at time $t$ (of cost $\alpha$) with the (at least) $\alpha$ cost of the adversary accessing the (at least) $\alpha$ requests to key $i$ between $t'$ and $t$.

\paragraph{Case ii: $w_i(t) \le 2\alpha$.}
From Claim \ref{clm:double}, $w_i(t)$ at least doubles between each adjustment that is caused by key $i$. So until
$w_i(t) \ge 2 \alpha$, key $i$ can cause at most $\log(\alpha)+1$ adjustments at a total adjustment cost that is less than $2\alpha \log \alpha$. 

It follows that the \emph{adjustment} cost, $\cost_A(\sigma)$, of Algorithm \ref{alg:A-BST} can be bounded as follows:
\begin{align}
    \cost_A(\sigma) &= \sum_{B_t \neq B_{t+1}} \alpha \\
    &= \sum_{i=1}^n \sum_{p_i < q_i/2} \alpha = \sum_{i=1}^n \left (\sum_{\substack{p_i(t) < q_i(t)/2 \\ w_i(t) < 2 \alpha}} \alpha + \sum_{\substack{p_i(t) < q_i(t)/2 \\ w_i(t) \ge 2 \alpha}} \alpha \right ) \label{eq:claim}\\
    &\le \sum_{i=1}^n \left ( 2\alpha \log \alpha + w_i \right) \label{eq:claim2}\\
    &= 2n \alpha \log \alpha + m \le 2m \label{eq:assumtion}
\end{align}
Equations \eqref{eq:claim} and \eqref{eq:claim2} follows from cases (i) and (ii) above. Eq. \eqref{eq:assumtion} follows form the assumption that $m\ge 2n \alpha \log \alpha$.

To finalize the proof we combine the above results to find the cost of A-BST.
\begin{align}
    \cost(\text{A-BST}, \sigma, B_1) &= \sum_{t=1}^m \depth(B_t,\sigma_t) + \alpha \mathbb{I}_{B_t \neq B_{t+1}} \\
        &= \cost_S(\sigma) + \cost_A(\sigma) \\
        &\le  4m + m H(\overline{W}) + 2m + 2m = m(8 + H(\overline{W}))
\end{align}
\qed
\end{proof}

\section{Discussion and Open Questions}
\label{sec:discussion} 
We believe that the matching model brings some interesting research directions. 
Our model can emulate the standard BST model when $\alpha=O(1)$, while dynamic optimally is a major and a long-standing open question for this case \cite{sleator1985self}, it may be possible that for some range of $\alpha$ the question becomes easier. Another property, possibly simpler, to prove or disprove approximation for, is a working-set like theorem \cite{sleator1985self}. 
A nice feature of the matching model is that it supports more complex networks than a BST.  
It is a future research direction to prove static optimality and other online properties for such networks.

Regarding A-BST, the algorithm can be extended in several directions to be more adaptive. First, it can work with a (memory) window instead of the whole history to become more dynamic. Careful attention is required to decide the window size in order not to lose the static optimally feature. More generally, as in arithmetic coding \cite{cover2012elements} which is an extension of the SFE coding, the prediction model can be independently replaced by a more sophisticated model than the current Laplace model.

\bibliographystyle{splncs04}
\bibliography{ref}

\end{document}